\documentclass[11pt]{amsart}
\usepackage{amssymb,color}
\numberwithin{equation}{section}
\usepackage{amssymb}
\usepackage{graphicx}
\usepackage{url}
\newcommand{\bea}{\begin{eqnarray}}
\newcommand{\eea}{\end{eqnarray}}
\newcommand{\ba}{\begin{array}}
\newcommand{\ea}{\end{array}}
\newcommand{\edc}{\end{document}}
\newcommand{\bc}{\begin{center}}
\newcommand{\ec}{\end{center}}
\newcommand{\be}{\begin{equation}}
\newcommand{\ee}{\end{equation}}

\newtheorem{thm}{Theorem}[section]

\newtheorem{defin}[thm]{Definition}
\theoremstyle{remark}
\newtheorem{rem}{Remark}[section]

\begin{document}
\title[Gibbs measures of the mixed Ising model]{Gibbs measures of the Ising model with mixed
spin-1 and spin-1/2 on a Cayley tree}
\author{Hasan Ak\i n$^{1,\dag}$}

\author{Farrukh Mukhamedov$^{2,3,\ddag}$}

\thanks{$^1$ \quad  The Abdus Salam International Centre for Theoretical Physics
(ICTP), Strada Costiera, 11, I - 34151 Trieste Italy\\
$^\dag$ \quad e-mail: {\tt
hakin@ictp.it; akinhasan25@gmail.com}\\
$^2$ \quad  Department of Mathematical Sciences, College of Science, The
United Arab Emirates University, P.O. Box, 15551, Al Ain Abu
Dhabi, UAE\\
$^3$ \quad  Department of Algebra and Analysis, Institute of Mathematics named after V.I.Romanovski,
4, University str., 100125, Tashkent, Uzbekistan\\
$^\ddag$ \quad e-mail: {\tt farrukh.m@uaeu.ac.ae; far75m@gmail.com
}}
\begin{abstract}
In the present paper, the Ising
model with mixed spin-(1,1/2) is considered on the second order Cayley tree. A construction of splitting Gibbs measures
corresponding the model is given which allows to establish the
existence of the phase transition (non-uniqueness of Gibbs
measures). We point out that, in the phase transition region, the
considered model has three translation-invariant Gibbs measures in
the ferromagnetic and anti-ferromagnetic regimes, while the
classical Ising model does not possesses such Gibbs
measures in the anti-ferromagnetic regime. It turns
out that the considered model, like the Ising model, exhibits a
disordered Gibbs measure. Therefore,
non-extremity and extremity of such disordered Gibbs measures is investigated by
means of tree-indexed Markov chains.
\\
\textbf{Keywords}: the mixed spin-(1,1/2) Ising model, Gibbs measures, phase transition, disorder phase.\\
\end{abstract}
\maketitle

\section{Introduction}

In the last decades, the Ising model has been one of the most
intensively studied model which used to describe critical
behaviours of certain systems in natural sciences. Many
interesting results have been observed in the phase transition
theory by means of exactly solvable Ising models \cite{Bax}.
Recently, several extension have been carried out in the a mixed
spin Ising model to describe a wide range of systems\cite{EAP16,EBM18,EKKK16}. Because it
is the lowest mixed-spin system, the mixed spin-1 and spin-1/2
system is an excellent candidate for studying mixed-spins. Higher
spins have higher critical temperatures, as is generally known,
i.e. spin-1 is a higher spin than spin-1/2. In addition, unlike
the spin-1/2 system, the spin-1 system exhibits first-order phase
transitions for acceptable parameters \cite{K88,Moraal}. In
\cite{Silva-Salinas1991}, Silva and Salinas performed some exact
calculations for a Curie-Weiss or mean-field version of the
mixed-spin Ising Hamiltonian on the Bethe lattice. In
\cite{Albayrak2020}, it has been studied the entropy and
isothermal entropy change of the mixed spin-(1,1/2) Ising model on
the Bethe lattice by means of the exact recursion relations for
the coordination numbers $q = 3,4$ and 6. In the present situation, there
are plenty of works devoted such types of investigations. All
existing works related to mixed spin Ising models deal with
physical approach from the physics perspective. Mostly, it has been
employed the exact iterative relations technique/the
renormalization group theory,  Monte Carlo simulations,
Green-function technique, the mean-field theory and a path
probability method (see \cite{GK20}, for related references).

On the other hand, it is essential to investigated the mentioned
models from a rigorous mathematical perspective. In this regard,
statistical mechanical approach is adequate, since it predicts the
relation between the observable macroscopic properties of the
system given only the knowledge of the microscopic interactions
between components. This can be explained by mathematical
framework. In this scheme, Gibbs measures are one of the central
objects of equilibrium statistical mechanics \cite{G}.
Therefore, one of the main problems of statistical physics is to
describe all Gibbs measures corresponding to the given
Hamiltonian. To the best knowledge of the authors, the mixed spin
Ising models on the Cayley tree (the same as Bathe lattice
\cite{Ost}) are  not well studied from the measure-theoretical point
of view. Therefore, one of the main aims of the present paper is
to develop a measure-theoretic approach (i.e. Gibbs measure
formalism) and rigorously establish the phase transition for the
Ising model with mixed spin on the Cayley trees.

In the rigorous approach, the Ising model on the Cayley tree
always has a disordered phase which corresponds to a Gibbs measure
having no external effect. In \cite{Bleher1990,Ioffe-1996}, it was
established that the disordered Gibbs measure, in the
ferromagnetic Ising model on the Cayley tree, is extreme, for
$T\geq T^{SG}_c$, where $T^{SG}_c$ is the critical temperature of
the spin glass model, and it is not extreme for $T< T^{SG}_c$. We
point out that the extremality of the disordered phase of lattice
models is important in the theory of information flows
\cite{Mar03,Mos01,MP03}. Therefore, there are several works
devoted to the extremality problem of Gibbs measures for variety
of models on trees (see
\cite{Khakimov-2016,KR17,Mukhamedov2022,Rahmatullaev2021,Roz,RKK18}).

In present paper, we consider the  Ising
model with mixed spin-(1,1/2) on the second order Cayley tree. We construct Gibbs measures
corresponding the model, which allows  us to establish the
existence of the phase transition (non-uniqueness of Gibbs
measures).

The main result of this paper is the following one (see Fig \ref{fig5}):

\begin{thm}\label{main} Assume that $J\neq 0$ and $\theta=\exp\{J\beta/2\}$. Then for the Ising model with mixed spins $(1,1/2)$ on the Cayley tree of order two, the following statements hold true:
\begin{enumerate}
\item[(i)]  If $\theta\in (0,\theta_{1})\cup(\theta_{2},\infty)$, where
$\theta_{1,2}$ are given by \eqref{TT}, then there occurs a phase transition;

 \item[(ii)] There exists $\tilde\theta_1$ and $\tilde\theta_2$ (see \eqref{tta}) such that the disordered phase $\mu_0$ is extreme if and only if
    $$
    \theta\in (\tilde\theta_1, \tilde\theta_2);
    $$
 \end{enumerate}
   \end{thm}

\begin{figure} [!htbp]\label{fig5}
\centering
\includegraphics[width=70mm]{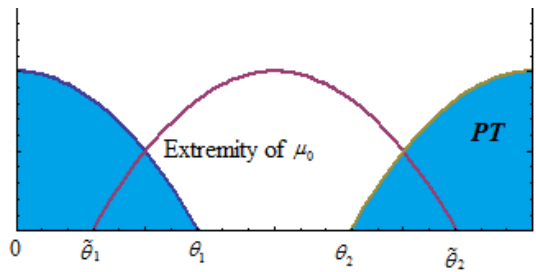}
\caption{Phase diagram}
\end{figure}

Interestingly, in the phase transition region of the considered model has a totally different phase diagram comparing to the Ising model. Indeed, the considered model possess three splitting translation-invariant Gibbs measures in both the ferromagnetic and
anti-ferromagnetic regimes, whereas the classical Ising
model does not have translation-invariant Gibbs measures in the
anti-ferromagnetic regime (i.e. $\theta <1$) \cite{Bax,G}.
Moreover, it turns out that the considered model, like the Ising model, exhibits a
disordered Gibbs measure.

The proof of this theorem will be given in forthcoming sections.

\section{Preliminaries and Splitting Gibbs measures}

Let $\Gamma^k = (V,L)$ be a semi-infinite Cayley tree of order
$k\geq 1$ with the root $x^0$ (whose each vertex has exactly $k+1$
edges, except for the root $x^0$, which has $k$ edges). Here $V$ is
the set of vertices and $L$ is the set of edges. The vertices $x$
and $y$ are called {\it nearest neighbors} and they are denoted by
$l=\langle{x,y}\rangle$ if there exists an edge connecting them. A collection of
the pairs $\langle{x,x_1}\rangle,\dots,\langle{x_{d-1},y}\rangle$ is called a {\it path} from
the point $x$ to the point $y$. The distance $d(x,y), x,y\in V$, on
the Cayley tree, is the length of the shortest path from $x$ to $y$.
$$
W_{n}=\left\{ x\in V\mid d(x,x^{0})=n\right\}, \ \
V_n=\bigcup\limits_{m=1}^{n}W_m , \ \
L_{n}=\left\{
l=<x,y>\in L\mid x,y\in V_{n}\right\}.
$$
The set of direct successors of $x$ is defined by
$$
S(x)=\left\{ y\in W_{n+1}:d(x,y)=1\right\}, x\in W_{n}.
$$
Observe that any vertex $x\neq x^{0}$ has $k$ direct successors and
$x^{0}$ has $k+1$.

Denote
\begin{eqnarray*}
\Gamma _+^k&=&\left\{x\in \Gamma ^2:d(x^{(0)},x)-even\right\}\\
\Gamma _-^k&=&\left\{x\in \Gamma ^2:d(x^{(0)},x)-odd\right\}
\end{eqnarray*}
In the present paper, we are going to consider spin state spaces as
follows: $\Phi=\{-1,0,+1\}$ and
$\Psi=\{-\frac{1}{2},+\frac{1}{2}\}$. The corresponding
configuration spaces are defined by $\Omega_+=\Phi^{\Gamma _+^k}$
and $\Omega_-=\Psi^{\Gamma _-^k}$. Moreover, we have
$\Omega_{+,n}=\Phi^{\Gamma _+^k\cap V_n}$ and
$\Omega_{-,n}=\Psi^{\Gamma _-^k\cap V_n}$. The configuration space
of the model is given by $\Xi=\Omega_+\times \Omega_-$. Elements
of $\Omega _+$, we denote by $\sigma(x)$, $x \in \Gamma _+^k$.
Similarly, elements of $\Omega_-$, we denote by $s(x)$, $x \in
\Gamma _-^k$.

Throughout in this paper, we always consider the mixed-spin model on the Cayley tree,
where configuration space in taken as $\Xi$. Namely, for the
configuration $\xi\in \Xi$, the associated sites
belong to successive generations of the tree (see Fig.
\ref{CayleyT01}). At odd-numbered levels of the tree, spins with
$\Psi$ elements are placed in the vertexes, while spins with
$\Phi$ elements are placed in the vertexes at even-numbered
levels, i.e.
$$
\xi (x)=\begin{cases}
 \sigma (x); & x\in \Gamma_+^k \\
 s(x); & x\in \Gamma_-^k,
\end{cases}
$$
where $\sigma\in \Phi=\{-1,0,+1\}$ and $s\in
\Psi=\{-\frac{1}{2},+\frac{1}{2}\}$.

\begin{figure} [!htbp]
\centering
\includegraphics[width=65mm]{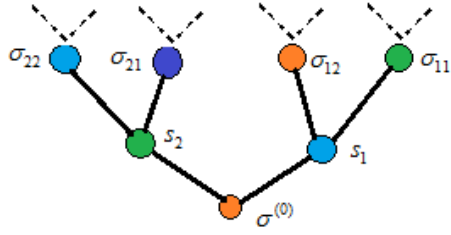}
\caption{Some generations of a second order Cayley tree of with a
$\sigma^{(0)}$ spin on in the root.}\label{CayleyT01}
\end{figure}
Now, let us consider the Ising model with mixed-spins $(1,1/2)$ given by the Hamiltonian
\begin{equation}\label{Hamil1}
H(\xi )=-J\sum _{\langle x,y\rangle } \xi(x)\xi(y),
\end{equation}
where $\xi\in \sigma\times s =\Xi.$

Let $\pmb{h}=(\pmb{h}_{\xi(x)}(x))_{x\in \Gamma^{2}}$, here
$$
\pmb{h}_{\xi (x)}(x)=\left\{
\begin{array}{ll}
 \pmb{h}_{\sigma (x)}, & x\in \Gamma _+; \\
 \pmb{\widetilde{h}}_{s(x)}, & x\in \Gamma _-.
\end{array}
\right.
$$
and  $\pmb{h}=(h_-,h_0,h_+)$,
$\pmb{\widetilde{h}}=(\widetilde{h}_{-\frac{1}{2}},\widetilde{h}_{\frac{1}{2}})$.

Now, for each $n\geq 1$, we define Gibbs measure $\mu _n^{\pmb{h}}$ by
\begin{equation}\label{Gibbs1}
\mu _n^{\pmb{h}}(\xi )=\frac{1}{Z_n}\exp \left\{-\beta H_n(\xi
)+\sum _{x\in W_n} \pmb{h}_{\xi (x)}(x)\right\},
\end{equation}
where $\xi\in \Xi_n:=\Omega_{+,n}\times \Omega_{-,n}.$

Recall that the sequence of measures
$\{\mu _n^{\pmb{h}}\}$ is compatible, if for all $n\geq 1$ and $\xi_{n-1}\in\Xi_{n-1}$ one hasx
\begin{equation}\label{compatile1}
\sum _{w\in \Xi ^{W_{n}}} \mu_n^{\pmb{h}}(\xi \lor
w)=\mu_{n-1}^{\pmb{h}}(\xi), \ \ \textrm{for all} \ \ n\geq 1,
\end{equation}
where
$$
\Xi^{W_n}=\left\{
\begin{array}{ll}
 \Phi ^{W_n}, & n-\text{even}; \\
 \Psi ^{W_n}, & n-\text{odd}.
\end{array}
\right.
$$
Here $\sigma_{n-1}\vee\omega_n$ is the concatenation of the configurations. In this setting, there is a unique measure $\mu$ on $\Omega$ such that for all $n$ and $\xi_n\in \Xi_n$
\begin{equation*}
\mu(\{\xi\ | \ V_n=\xi_n\})=\mu _n^{\pmb{h}}(\xi_n).
\end{equation*}
Such a measure is called a \textit{splitting Gibbs measure (SGM)} corresponding to the model \cite{Roz}.

The following result describes the condition on $\pmb{h}$ ensuring that the sequence $\{\mu
_n^{\pmb{h}}\}$ is compatible.

\begin{thm}\label{thm-comp} The sequence of measures $\{\mu
_n^{\pmb{h}}\}$, $n=1,2,...$ given by \eqref{compatile1} is compatible
if and only if  for any $x\in V$ the following equations hold:
\begin{equation}\label{compatible1a}
\exp(\pmb{h}_{1}(x)-\pmb{h}_0(x))=\prod _{y\in
S(x)}\left(\frac{\exp (-\frac{J\beta
}{2}+\widetilde{h}_{-\frac{1}{2}}(y)+ \exp (\frac{J\beta
}{2}+\widetilde{h}_{\frac{1}{2}}(y))}{\exp
(\widetilde{h}_{-\frac{1}{2}}(y)+ \exp
(\widetilde{h}_{\frac{1}{2}}(y))}\right),
\end{equation}
\begin{equation}\label{compatible1b}
\exp(\pmb{h}_{-1}(x)-\pmb{h}_0(x))=\prod _{y\in
S(x)}\left(\frac{\exp (\frac{J\beta
}{2}+\widetilde{h}_{-\frac{1}{2}}(y)+ \exp (-\frac{J\beta
}{2}+\widetilde{h}_{\frac{1}{2}}(y)))}{\exp
(\widetilde{h}_{-\frac{1}{2}}(y)+ \exp
(\widetilde{h}_{\frac{1}{2}}(y))}\right),
\end{equation}
    \begin{equation}\label{even-case1}\exp
(\pmb{\widetilde{h}}_{\frac{1}{2}}(x)-\pmb{\widetilde{h}}_{-\frac{1}{2}}(x))
=\prod _{y\in S(x)}\left(\frac{\exp(\frac{-J\beta}{2} +
\pmb{h}_{-1}(y))+\exp(\pmb{h}_{0}(y))+\exp(\frac{J\beta}{2} +
\pmb{h}_{1}(y))}{\exp(\frac{J\beta}{2} +
\pmb{h}_{-1}(y))+\exp(\pmb{h}_{0}(y))+\exp(\frac{-J\beta}{2} +
\pmb{h}_{1}(y))}\right).
\end{equation}

\end{thm}
\begin{proof} \emph{Necessary}. Let us consider two cases with
respect $n$;\\
(1) $n$-odd case: assume that the eq. \eqref{compatile1} is
satisfied; Then
\begin{equation}\label{compatile2}
\exp \left(\sum _{x\in W_{n-1}} \pmb{h}_{\xi
(x)}(x)\right)=\frac{Z_{n-1}}{Z_n}\sum _{w_n\in \Psi^{W_n}} \exp
\left(\sum _{x\in W_{n-1}} \sum _{y\in S(x)} (J\beta \xi
_{n-1}(x)w_n(y)+\pmb{\widetilde{h}}_{\xi (y)}(y))\right).
\end{equation}
Because $n$ is odd, from eq. \eqref{compatile2}, we get
\begin{equation}\label{compatile3}
\prod_{x\in W_{n-1}} \exp \left(\pmb{h}_{\xi
(x)}(x)\right)=\frac{Z_{n-1}}{Z_n}\sum _{s\in \Psi^{W_n}} \prod
_{x\in W_{n-1}} \prod _{y\in S(x)} \exp \left(J\beta
\sigma(x)s(y)+ \pmb{\widetilde{h}}_{s(y)}(y)\right).
\end{equation}
Fix $x\in W_{n-1}$ and take into account configurations
$\sigma_{n-1}=\overline{\sigma}_{n-1}$ and
$\sigma_{n-1}=\widetilde{\sigma}_{n-1}$ on $\Phi^{W_{n-1}}$ which
coincide on $\Phi^{W_{n-1}\setminus \{x\}}$ and we rewrite the eq.
\eqref{compatile3} for $\overline{\sigma}_{n-1}(x)=1$,
$\overline{\overline{\sigma}}_{n-1}(x)=0$ and
$\overline{\overline{\overline{\sigma}}}_{n-1}(x)=-1.$ Therefore, one finds
\begin{eqnarray*}
(i)&&\exp (\pmb{h}_{-1}(x))=\prod_{y\in S(x)} \sum _{j\in
\left\{-\frac{1}{2},\frac{1}{2}\right\}} \exp \{-J\beta j+ \pmb{\widetilde{h}}_j(y)\},\\
(ii)&& \exp (\pmb{h}_{0}(x))=\prod _{y\in S(x)} \sum _{j\in
\left\{-\frac{1}{2},\frac{1}{2}\right\}} \exp \{\pmb{\widetilde{h}}_j(y)\},\\
(iii)&& \exp (\pmb{h}_{1}(x))=\prod _{y\in S(x)} \sum _{j\in
\left\{-\frac{1}{2},\frac{1}{2}\right\}} \exp \{J\beta j+
\pmb{\widetilde{h}}_j(y)\}.
\end{eqnarray*}
From the last ones, we have
\begin{equation}\label{compatile3a}
\exp (\pmb{h}_{-1}(x))=\prod _{y\in S(x)} \left(\exp (\frac{J\beta
}{2}+\pmb{\widetilde{h}}_{-\frac{1}{2}}(y)+ \exp (-\frac{J\beta
}{2}+\pmb{\widetilde{h}}_{\frac{1}{2}}(y))\right),
\end{equation}
\begin{equation}\label{compatile3b}
\exp (\pmb{h}_{1}(x))=\prod _{y\in S(x)} \left(\exp (-\frac{J\beta
}{2}+\pmb{\widetilde{h}}_{-\frac{1}{2}}(y)+ \exp (\frac{J\beta
}{2}+\pmb{\widetilde{h}}_{\frac{1}{2}}(y))\right),
\end{equation}
\begin{equation}\label{compatile3c}
\exp (\pmb{h}_{0}(x))=\prod _{y\in S(x)}  \left(\exp
(\pmb{\widetilde{h}}_{-\frac{1}{2}}(y)+ \exp
(\pmb{\widetilde{h}}_{\frac{1}{2}}(y))\right).
\end{equation}
Therefore, dividing the equations \eqref{compatile3a} and
\eqref{compatile3b} by  \eqref{compatile3c}, respectively, we obtain
\begin{eqnarray*}
\exp(\pmb{h}_{-1}(x)-\pmb{h}_0(x))&=&\prod _{y\in
S(x)}\left(\frac{\exp (\frac{J\beta
}{2}+\pmb{\widetilde{h}}_{-\frac{1}{2}}(y)+ \exp (-\frac{J\beta
}{2}+\pmb{\widetilde{h}}_{\frac{1}{2}}(y)))}{\exp
(\pmb{\widetilde{h}}_{-\frac{1}{2}}(y)+ \exp
(\pmb{\widetilde{h}}_{\frac{1}{2}}(y))}\right),\\
\exp(\pmb{h}_{1}(x)-\pmb{h}_0(x))&=&\prod _{y\in
S(x)}\left(\frac{\exp (-\frac{J\beta
}{2}+\pmb{\widetilde{h}}_{-\frac{1}{2}}(y)+ \exp (\frac{J\beta
}{2}+\pmb{\widetilde{h}}_{\frac{1}{2}}(y))}{\exp
(\pmb{\widetilde{h}}_{-\frac{1}{2}}(y)+ \exp
(\pmb{\widetilde{h}}_{\frac{1}{2}}(y))}\right).
\end{eqnarray*}
(2) \emph{The case $n$ even}:  This case is proceeded by the same argument as above.  For any configuration $\sigma_{n-1}\in
\Xi_{n-1}$,
\begin{equation}\label{even-case2} \exp \left(\sum _{x\in W_{n-1}}
\pmb{\widetilde{h}}_{\xi (x)}(x)\right)=\frac{Z_{n-1}}{Z_n}\sum
_{w_n\in \Phi^{W_n}} \exp \left(\sum _{x\in W_{n-1}} \sum _{y\in
S(x)} (J\beta \sigma _{n-1}(x)w_n(y)+\pmb{h}_{\xi (y)}(y))\right).
\end{equation}
From eq. \eqref{even-case2} and due to $n$ even, we get
\begin{equation}\label{even-case3}
\prod_{x\in W_{n-1}} \exp \left(\pmb{\widetilde{h}}_{\xi
(x)}(x)\right)=\frac{Z_{n-1}}{Z_n}\sum _{w_n\in \Phi^{W_n}} \prod
_{x\in W_{n-1}} \prod _{y\in S(x)} \exp \left(J\beta \sigma
_{n-1}(x)w_n(y)+ \pmb{h}_{\sigma(y)}(y)\right).
\end{equation}
Fix $x\in W_{n-1}$ and take into account configurations
$s_{n-1}=\overline{s}_{n-1}$ and $s_{n-1}=\widetilde{s}_{n-1}$ on
$\Psi^{W_{n-1}}$ which coincide on $\Psi^{W_{n-1}\setminus \{x\}}$ and we
rewrite the eq. \eqref{even-case3} for
$\overline{s}_{n-1}(x)=\frac{1}{2}$ and
$\overline{\overline{s}}_{n-1}(x)=-\frac{1}{2}$ as follows;
\begin{eqnarray*}
(iv)\exp (\pmb{\widetilde{h}}_{\frac{1}{2}}(x))&=&\prod _{y\in
S(x)}\sum _{j\in \Phi} \exp\left(\frac{1}{2} J\beta  j+
\pmb{h}_j(y)\right),\\ (v) \exp
(\pmb{\widetilde{h}}_{-\frac{1}{2}}(x))&=&\prod _{y\in S(x)}\sum
_{j\in \Phi} \exp\left(-\frac{1}{2} J\beta   j+
\pmb{h}_j(y)\right).
\end{eqnarray*}
From the equations (iv) and (v), we have
\begin{eqnarray}\label{even-case4}
\exp (\pmb{\widetilde{h}}_{\frac{1}{2}}(x))&=&\prod _{y\in
S(x)}\left(\exp(\frac{-J\beta}{2} +
\pmb{h}_{-1}(y))+\exp(\pmb{h}_{0}(y))+\exp(\frac{J\beta}{2} +
\pmb{h}_{1}(y))\right),
\end{eqnarray}
\begin{eqnarray}\label{even-case5}
\exp (\pmb{\widetilde{h}}_{-\frac{1}{2}}(x))&=&\prod _{y\in
S(x)}\left(\exp(\frac{J\beta}{2} +
\pmb{h}_{-1}(y))+\exp(\pmb{h}_{0}(y))+\exp(\frac{-J\beta}{2} +
\pmb{h}_{1}(y))\right).
\end{eqnarray}
Dividing \eqref{even-case4} by \eqref{even-case5}, we find
\begin{equation}\label{even-case6}
\exp
(\pmb{\widetilde{h}}_{\frac{1}{2}}(x)-\pmb{\widetilde{h}}_{-\frac{1}{2}}(x))
=\prod _{y\in S(x)}\left(\frac{\exp(\frac{-J\beta}{2} +
\pmb{h}_{-1}(y))+\exp(\pmb{h}_{0}(y))+\exp(\frac{J\beta}{2} +
\pmb{h}_{1}(y))}{\exp(\frac{J\beta}{2} +
\pmb{h}_{-1}(y))+\exp(\pmb{h}_{0}(y))+\exp(\frac{-J\beta}{2} +
\pmb{h}_{1}(y))}\right).
\end{equation}

\emph{Sufficiency}: Assume that \eqref{Gibbs1} is satisfied. Then
\begin{equation}\label{Sufficiency1}
A(x)h_{\sigma (x)}=\prod _{y\in S(x)} \sum _{u\in \Psi} \exp
(J\beta \sigma (x)u+\pmb{\widetilde{h}}_u),\ \ \sigma (x)\in \Phi
\end{equation}
for some function $A(x)$, $x\in V$. So, \small
\begin{eqnarray}\label{Sufficiency2}
 \text{RHS of}\ \eqref{Gibbs1}=\frac{\exp \left(-\beta H_{n-1}(\sigma _{n-1})\right)}{Z_n}
 \prod\limits _{x\in W_{n-1}} \prod\limits_{y\in S(x)} \sum\limits_{\eta (y)\in \Psi}
 \exp \left(J\beta \sigma _{n-1}(x)\eta (y)+\pmb{\widetilde{h}}_{\eta
 (y)}\right).
\end{eqnarray}
\normalsize
Substituting \eqref{Sufficiency1} into \eqref{Sufficiency2} and
denoting $\mathcal{A}_{n}(x)=\prod\limits _{x\in W_{n-1}}A(x)$, one has
\begin{equation}\label{Sufficiency3}
\text{RHS\ of}\
\eqref{Sufficiency1}=\frac{\mathcal{A}_{n-1}}{Z_n}\exp (-\beta
H_{n-1}(\sigma _{n-1})\prod _{x\in W_{n-1}} \exp (h_{\sigma
_{n-1}(x)}).
\end{equation}
Since $\mu^{(n)}$ is a probability measure, then
$$
\sum _{\sigma _{n-1}\in \Xi ^{V_{n-1}}} \sum _{w_n\in \Psi^{W_n}}
\mu ^{(n)}(\sigma _{n-1}\lor w_n)=1.
$$
Hence from \eqref{Sufficiency3}, we find
$Z_{n-1}\mathcal{A}_{n-1}=Z_n$ and the eq. \eqref{Gibbs1} is
satisfied.
\end{proof}

\begin{rem} It is stressed that the recursive formula $Z_{n-1}\mathcal{A}_{n-1}=Z_n$ allows us to investigate free energies of the model \cite{Akin2022,MAKfree2017}. However, such kind of study will be reported elsewhere.
\end{rem}

\section{Splitting Translation invariant Gibbs measures (TIGMs)}

In this section, we deal with the existence of splitting
translation-invariant Gibbs measures (STIGMs) corresponding to the
Ising model with mixed spin-1 and spin-1/2 by analyzing the the
equations \eqref{compatible1a}-\eqref{even-case1}. Note that the
vector valued functions
$\pmb{\widetilde{h}}=\{\pmb{\widetilde{h}}_{-\frac{1}{2}}(x),\pmb{\widetilde{h}}_{\frac{1}{2}}(x)\}$
and $\pmb{h}=\{\pmb{h}_{-1}(x),\pmb{h}_{0}(x),\pmb{h}_{1}(x)\}$
are considered as \textit{translation-invariant} if
$\pmb{\widetilde{h}}_{i}(x)=\pmb{\widetilde{h}}_{i}(y)$ and
$\pmb{h}_{j}(x)=\pmb{h}_{j}(y)$ for all $y \in S(x)$, where
$i\in\{-\frac{1}{2}, \frac{1}{2}\}$ and $j\in \{-1,0 1\}$ (see
\cite{Roz} for details). Then the corresponding measure is called \textit{splitting translation invariant Gibbs measure}.
In what follows, we restrict ourselves to the tree of order two (i.e. $k=2$).

For the sake of convenience, let us denote: $\pmb{h}_{j}:=\pmb{h}_{j}(x)$ for all $x\in \Gamma_+^2$, $j\in \Phi$, $\pmb{\widetilde{h}}_{i}:=
\pmb{\widetilde{h}}_{i}(x)$, $x\in \Gamma_-^2$, $i\in
\Psi$.

Assuming $U_1=\pmb{h}_{-1}-\pmb{h}_0; U_2=\pmb{h}_1-\pmb{h}_0$
and $V=\pmb{\widetilde{h}}_{1/2}-\pmb{\widetilde{h}}_{-1/2}$, from the equations \eqref{compatible1a}, \eqref{compatible1b}, \eqref{even-case1}
we obtain
\begin{equation}\label{TI-1}
 \exp(U_1)=\left(\frac{\theta^2+ \exp(V)}{\theta (1+  \exp(V))}\right)^{2},
 \end{equation}
 \begin{equation}\label{TI-2}
  \exp(U_2)=\left(\frac{1 + \theta^2  \exp(V)}{\theta (1+ \exp(V)}\right)^{2}.
\end{equation}
 \begin{equation}\label{TI-3}
 \exp (V)=\left(\frac{\exp (U_1)+\theta ^2\exp (U_2)+\theta }
 {\theta ^2\exp (U_1)+\exp (U_2)+\theta }\right)^2,
\end{equation}
where $\theta=e^{{\frac{J\beta}{2}}}$.

Now, let us consider substitutions as $X=\exp(U_1)$, $Y=\exp(U_2)$ and $Z=\exp(V)$, then the last equations
reduce to
\begin{eqnarray}\label{dynamicals1}
 X&=&\left(\frac{\theta ^2+Z}{\theta (1+Z)}\right)^2,\\
 \label{dynamicals2}
 Y&=&\left(\frac{1+\theta ^2Z}{\theta (1+Z)}\right)^2,\\\label{dynamicals3}
 Z&=&\left(\frac{X+\theta ^2Y+\theta }{\theta ^2X+Y+\theta }\right)^2.
\end{eqnarray}
By Theorem \ref{thm-comp}, we infer that fixed points of the obtained system of equations determine TIGM corresponding to the Ising model with mixed spin-1 and spin-1/2.

We notice that the \eqref{dynamicals1}-\eqref{dynamicals3} has a solution  $Z_0=1$, $X_0=Y_0=\left(\frac{\theta ^2+1}{2\theta }\right)^2.$ The corresponding STIGM $\mu_0$ is called \textit{disordered phase} of the model. This means that the corresponding STIGM does not have external boundary conditions. Moreover, it exists for any value of the temperature $T=1/\beta$.

\begin{rem} We notice that for symmetric models (interactions) the disordered Gibbs measure corresponds (in most cases)
to the zero solution of the recurrent equations. However, if the model is not symmetric, then
corresponding recurrent equation has no zero solution. In that case, the disordered
measure is defined as a translation-invariant Gibbs measure which corresponds
to the free boundary condition. Therefore, in some literatures (see for example
\cite{Mos01}) the disordered measure is called free measure.
\end{rem}

\section{Stability analysis of the dynamical system}

In this section, we are going to investigate stability of the fixed point
$Z_0=1$, $X_0=Y_0=\left(\frac{\theta ^2+1}{2\theta }\right)^2.$ It is clear to calculate the Jacobian of the fixed point:
\begin{equation}\label{JakobeM}
J_F:=\left(
\begin{array}{ccc}
 0 & 0 & -\frac{(\theta ^2-\theta ^{-2})}{4} \\
 0 & 0 & \frac{(\theta ^2-\theta ^{-2})}{4} \\
 -\frac{8\theta ^2(\theta ^2-1)}{1+3\theta ^2+4\theta ^3+3\theta ^4+\theta ^6} &
 \frac{8\theta ^2(\theta ^2-1)}{1+3\theta ^2+4\theta ^3+3\theta ^4+\theta ^6} & 0
\end{array}
\right).
\end{equation}
The eigenvalues of the matrix $J_F$ are given as follows:
\begin{center}
$\lambda_1=0,\lambda_2=-\frac{2 \sqrt{1-\theta ^2-\theta ^4+\theta
^6}}{\sqrt{1+3 \theta ^2+4 \theta ^3+3 \theta ^4+\theta
^6}},\lambda_3=\frac{2 \sqrt{1-\theta ^2-\theta ^4+\theta
^6}}{\sqrt{1+3 \theta ^2+4 \theta ^3+3 \theta ^4+\theta ^6}}$.
\end{center}
In order to characterize the behavior of the dynamical system, one needs investigate the eigenvalues of $J_F$ in \eqref{JakobeM}.
\begin{figure} [!htbp]\label{fig2ab1a}
\centering
\includegraphics[width=80mm]{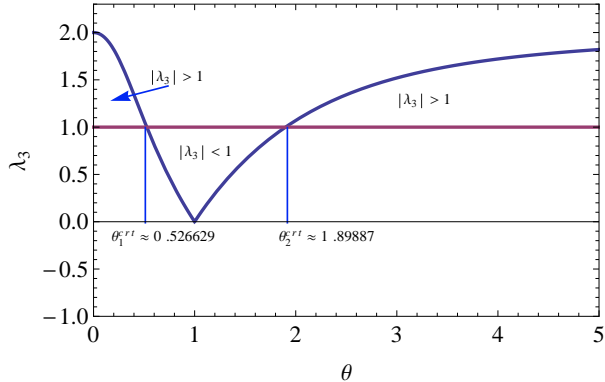}
\caption{The graph of the eigenvalue $\lambda_3$ for
$\theta>0$.}\label{fig2ab1a}
\end{figure}

From $|\lambda_2|=|\lambda_3|=\frac{2 \sqrt{1-\theta ^2-\theta ^4+\theta
^6}}{\sqrt{1+3 \theta ^2+4 \theta ^3+3 \theta ^4+\theta ^6}}=1$,
we have $ \theta _2^{crt}\approx 1.89887$ and  $\theta
_1^{crt}\approx
 0.526629$. For $\theta\in (\theta _1^{crt},\theta _2^{crt})$, one gets
$|\lambda_3|<1$ and for $\theta\in (0,\theta _1^{crt})\cup(\theta
_2^{crt},\infty)$, one has $|\lambda_3|>1$ (see Fig.
\ref{fig2ab1a}).

\begin{rem} If $\theta\in (0,\theta _1^{crt})\cup(\theta _2^{crt},\infty)$, since the fixed point is repelling,
this indicates that the function has other fixed points.
Therefore, it predicts the existence of a phase transition. We will
explain our prediction in the next sections.
\end{rem}

\section{The existence of a phase transition}

In this section, we are going to prove the first part of Theorem \ref{main}. We recall that the phase transition problem is one of the main challenges of statistical
mechanics \cite{G}. The existence of more than one Gibbs measure
corresponding to a lattice model indicated the phase transition
(see \cite{Bax,G, Roz} for details). In this section,
we investigate the existence and non-uniqueness of STIGM. To do so, we investigate the fixed points of the system
\eqref{dynamicals1}-\eqref{dynamicals3}.

If we substitute the equations \eqref{dynamicals1} and
\eqref{dynamicals2} into \eqref{dynamicals1}, then one finds
\begin{equation}\label{0Fixed-1a}
Z=F(Z)
\end{equation}
where
\begin{equation}\label{Fixed-1a}
F(Z)=\frac{\left(Z^2+\theta ^2+2 Z \theta ^2+\theta ^3+2 Z \theta
^3+Z^2 \theta ^3+\theta ^4+2 Z \theta ^4+Z^2 \theta
^6\right)^2}{\left(1+2 Z \theta ^2+Z^2 \theta ^2+\theta ^3+2 Z
\theta ^3+Z^2 \theta ^3+2 Z \theta ^4+Z^2 \theta ^4+\theta
^6\right)^2}.
\end{equation}
Let us analyze the equation $Z=F(Z)$.

After some tricky operations, \eqref{0Fixed-1a} can be reduced to

\begin{equation}\label{fixed1}
(Z-1)(AZ^4 +BZ^3 + CZ^2+ BZ+A)=0,
\end{equation}
where
\begin{eqnarray*}
A&=&\theta ^4 \left(1+\theta +\theta ^2\right)^2,\\
B&=&\left(-1-2 \theta ^3+5 \theta ^4+10 \theta ^5+12 \theta ^6+10
\theta ^7+5 \theta ^8-2 \theta ^9-\theta ^{12}\right),\\
C&=&\left(-1-2 \theta ^2-4 \theta ^3+7 \theta ^4+16 \theta ^5+22
\theta ^6+16 \theta ^7+7 \theta ^8-4 \theta ^9-2 \theta
^{10}-\theta ^{12}\right).
\end{eqnarray*}
By dividing the eq. \eqref{fixed1} to $Z^2$, one gets
$$
A(Z^2+\frac{1}{Z^2})+B(Z+\frac{1}{Z})+C=0.
$$
Denoting
\begin{equation}\label{root1}
t=Z+\frac{1}{Z},\end{equation}
we then obtain
\begin{equation}\label{fixed2}
f(t):=At^2 +Bt + (C-2A)=0.
\end{equation}
The solutions of the last equation are:
\begin{equation}\label{fixed2a}
t_1=\frac{-B-\sqrt{8 A^2+B^2-4 A C}}{2 A},t_2=\frac{-B+\sqrt{8
A^2+B^2-4 A C}}{2 A}
\end{equation}

Taking into account $t=Z+\frac{1}{Z}>2$, there are two
cases which are interesting for us: $t_1<2<t_2$ or $2<t_1<t_2$.

\textbf{Case I}.  Assume that $t_1<2<t_2$. Then, for the parabola
$f(t)$ given in \eqref{fixed2}, it should be satisfied the
inequality $Af(2)<0$. Due to $A>0$, we infer that
$$
2(A+B)+C<0.
$$
After some algebraic operations, the last one reduces to
\begin{equation}\label{fixed3}
3-7 \theta ^2-4 \theta ^3-7 \theta ^4+3 \theta ^6>0.
\end{equation}
Denoting $\rho= \theta+\frac{1}{\theta}$, the inequality \eqref{fixed3} can be rewritten as follows
\begin{equation}\label{fixed4}
\left\{
\begin{array}{c}
 g(\rho):=3\rho^3-16\rho-4>0 \\
 \rho>2, \ (\text{since} \ J\neq 0).
\end{array}
\right.
\end{equation}

One can see that (see
Fig. \ref{fig3a1}) the solution of the last inequality is $\rho\in (\rho_{crt},\infty)$, here $\rho_{crt}\approx 2.4255$.
Due to $\theta+\frac{1}{\theta}>\rho_{crt}$ we obtain
\begin{equation}\label{fixed5}
\theta^2-2.4255\theta+1>0
\end{equation}
which has solution $\theta\in
(0,\theta_{1})\cup(\theta_{2},\infty)$, where
\begin{equation}\label{TT}
\theta_{1,2}=\frac{2.4255\pm \sqrt{(2.4255)^2-4}}{2}.
\end{equation}
Note that $\theta_{1}\approx 0.52663$ and
$\theta_{2}\approx 1.899$. 

\begin{figure} [!htbp]\label{fig3a1}
\centering
\includegraphics[width=70mm]{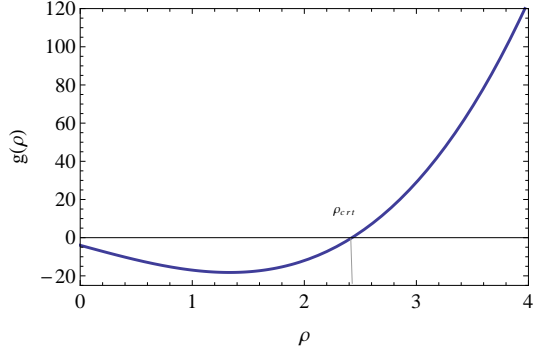}
\caption{The graph of $g(\rho)$ given in eq.
\eqref{fixed4}.}\label{fig3a1}
\end{figure}

In the considered case, we take $t_2$ as the root. Due to
condition \eqref{root1}, one gets $Z^2-t_2Z+1=0$. Therefore, there are two roots
$$
Z_1^*=\frac{t_2+\sqrt{t_2^2-4}}{2},Z_2^*=\frac{t_2-\sqrt{t_2^2-4}}{2}.
$$

The corresponding measures we denote by $\mu_{1,2}$. Consequently, if $\theta\in (0,\theta_{1})\cup(\theta_{2},\infty)$, there occurs a phase transition.

\textbf{Case 2}. Assume that $2<t_1<t_2$. It is enough to consider $\theta_1<\theta<\theta_2$ (see Case 1).  Then, for the parabola $f(t)$ in
\eqref{fixed2}, from $A>0$ it should be satisfied the inequality
$f(2)>0$.

Consequently, $-\frac{B}{2A}>2$, 
that is $B+4A<0$. After some algebraic
operations, one finds
\begin{equation*}\label{fixed3}
-\left(1-3 \theta ^2-2 \theta ^3-3 \theta ^4+\theta ^6\right)
\left(1+3 \theta ^2+4 \theta ^3+3 \theta ^4+\theta ^6\right)<0.
\end{equation*}
So,
\begin{equation}\label{fixed3}
h(\theta):=1-3\theta ^2-2\theta ^3-3\theta ^4+\theta
^6>0.
\end{equation}
\begin{figure} [!htbp]\label{fig3a}
\centering
\includegraphics[width=70mm]{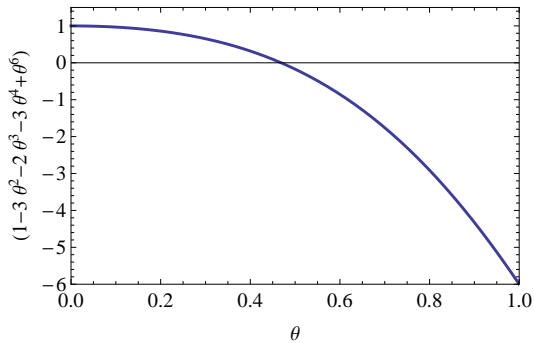}
\caption{The graph of $h(\theta)$ given in eq.
\eqref{fixed3}.}\label{fig3a}
\end{figure}
One can see that $h(\theta)>0$, if $0<\theta<0.452$. Therefore,
due to $\theta>\theta_1$ ($\theta_1=0.5$), we don't have any
solution in the current situation (see Fig. \ref{fig3a}).

\begin{figure} [!htbp]\label{fig2}
\centering
\includegraphics[width=70mm]{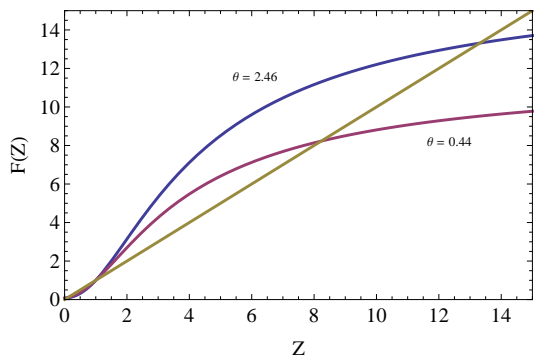}
\caption{The graphs of $F(Z)$ given in \eqref{Fixed-1a} for
$\theta=2.46$ and $\theta=0.44$.}\label{fig2}
\end{figure}

Hence, we can formulate the following result.
\begin{thm} Assume that $J\neq 0$ and $\theta=\exp\{J\beta/2\}$. If $\theta\in (0,\theta_{1})\cup(\theta_{2},\infty)$, where
$\theta_{1,2}$ are given by \eqref{TT}, then for the Ising model with mixed spins $(1,1/2)$ there is a phase transition.
\end{thm}

\begin{rem} We stress that the considered model has a totally different phase diagram comparing to the Ising model. Indeed, the considered model possess three STIGM in both the ferromagnetic and
anti-ferromagnetic regimes, whereas the classical Ising
model does not have translation-invariant Gibbs measures in the
anti-ferromagnetic regime (i.e. $\theta <1$) \cite{Bax,G} (see Fig. \ref{fig2}).
\end{rem}

\section{Extremality of the Disordered Phase by Tree-indexed Markov chains}

In this section, we are going to prove the second part of Theorem \ref{main}. The extremality of the disordered phase of the Ising model with mixed spin $(1,1/2)$ is investigated by means of tree-indexed Markov chains. Let us recall tree-indexed Markov chains. Suppose that we are given a tree with vertices set $V$,  a probability measure $\nu$ and a transition matrix $\mathbb{P}=(P_{ij})_{i,j\in \Phi}$ on the single-site space which is here a finite set $\Phi=\{1,2,\dots,q\}$. We can obtain a tree-indexed Markov chain $X:V\rightarrow \Phi$ by choosing $X(x^0)$ according to $\nu$ and choosing $X{\nu},$ for each vertex $v \neq x^{0}$, using the transition probabilities given the value of its parent, independently of everything else (see  \cite[Definition 12.2]{G} for details).

Let us define the entries of the transition probabilities matrices as follows:
$\mathbb{P}= (P_{ij})$ as
$$
P_{ij}=\frac{\exp(ij\beta
J+\pmb{\widetilde{h}}_{j})}{\sum\limits_{u=\mp
\frac{1}{2}}\exp(iu\beta J+\pmb{\widetilde{h}}_{u})},
$$
where $i\in \{-1,0,1\}$ and $j\in \{-\frac{1}{2},\frac{1}{2}\}$.

Assume that
$Z=\exp(\pmb{\widetilde{h}}_{\frac{1}{2}}-\pmb{\widetilde{h}}_{-\frac{1}{2}})$, then
\begin{eqnarray}\label{stochasticP1}
\mathbb{P}=\left(
\begin{array}{cc}
 \frac{\theta ^2Z}{1+\theta ^2Z} & \frac{1}{1+\theta ^2Z} \\
 \frac{Z}{1+Z} & \frac{1}{1+Z} \\
 \frac{Z}{\theta ^2+Z} & \frac{\theta ^2}{\theta ^2+Z}
\end{array}
\right).
\end{eqnarray}

Similarly, we define the entries of the transition probabilities
matrix $\mathbb{Q}= (Q_{ij})$ by
$$
Q_{ij}=\frac{\exp(ij\beta J+\pmb{h}_{j})}{\sum\limits_{\ell\in
\{-1,0,1\}}\exp(i\ell\beta J+\pmb{h}_{\ell})},
$$
where $i\in \{-\frac{1}{2},\frac{1}{2}\}$ and $j\in \{-1,0,1\}$.

Assuming $X=\exp(\pmb{h}_{-1}-\pmb{h}_{0})$ and
$Y=\exp(\pmb{h}_{1}-\pmb{h}_{0})$, then one finds
\begin{eqnarray}\label{stochasticQ1}
\mathbb{Q}=\left(
\begin{array}{ccc}
 \frac{\theta ^2X}{\theta ^2X+\theta +Y} & \frac{\theta }{\theta ^2X+\theta +Y} & \frac{Y}{\theta ^2X+\theta +Y} \\
 \frac{X}{X+\theta +\theta ^2Y} & \frac{\theta }{X+\theta +\theta ^2Y} & \frac{\theta ^2Y}{X+\theta +\theta ^2Y}
\end{array}
\right).
\end{eqnarray}

Therefore, from \eqref{stochasticP1} and \eqref{stochasticQ1}, a
STIGM corresponding to a vector $v = (X, Y,Z)\in \mathbb{R}^{3}$
(which is the solution to system of the equations
\eqref{dynamicals1}-\eqref{dynamicals3}) is a tree-indexed Markov
chain with states $\{-1,0, 1\}$ and transition probabilities
matrix is obtained as follows:

\begin{equation*}\label{matHS1}
\mathbb{H}=\mathbb{P}\mathbb{Q}=\frac{1}{(X+\theta +Y\theta
^2)}\left(
\begin{array}{ccc}
 \frac{X\left(1+\theta ^4\sqrt{Z^3}\right)}{\left(1+Z\theta ^2\right)} & \frac{\theta \left(1+\theta ^2\sqrt{Z^3}\right)}{\left(1+Z\theta ^2\right)} & \frac{Y\theta ^2\left(1+\sqrt{Z^3}\right)}{\left(1+Z\theta ^2\right)} \\
 \frac{X\left(1+\theta ^2\sqrt{Z^3}\right)}{(1+Z)} & \frac{\theta \left(1+\sqrt{Z^3}\right)}{(1+Z)} & \frac{Y\left(\theta ^2+\sqrt{Z^3}\right)}{(1+Z)} \\
 \frac{X\theta ^2\left(1+\sqrt{Z^3}\right)}{\left(Z+\theta ^2\right)} & \frac{\theta \left(\theta ^2+\sqrt{Z^3}\right)}{\left(Z+\theta ^2\right)} & \frac{Y\left(\theta ^4+\sqrt{Z^3}\right)}{\left(Z+\theta ^2\right)}
\end{array}
\right).
\end{equation*}
Note that the matrices $\mathbb{P},\mathbb{Q}$ and $\mathbb{H}$
are stochastic ones. Noticing that the disordered phase corresponds to
$Z=1$ and $X=Y=\left(\frac{1+ \theta ^2}{2\theta }\right)^2$, then the corresponding matrix $\mathbb{H}$ reduces to
\begin{equation}\label{matHS2}
\mathbb{H}=\frac{1}{1+3 \theta ^2+4 \theta ^3+3 \theta ^4+\theta
^6}\left(
\begin{array}{ccc}
 \left(1+\theta ^2\right) \left(1+\theta ^4\right) & 4 \theta ^3 & 2 \theta ^2 \left(1+\theta ^2\right) \\
 \frac{\left(1+\theta ^2\right)^3}{2 } & 4 \theta ^3 & \frac{\left(1+\theta ^2\right)^3}{2 } \\
 2 \theta ^2 \left(1+\theta ^2\right) & 4 \theta ^3 & \left(1+\theta ^2\right) \left(1+\theta ^4\right)
\end{array}
\right).
\end{equation}

In order to investigate the criterion of extremality of a STISGM,
let us now give the definitions of quantities $\kappa$ and
$\gamma$ given in the Ref. \cite{MSW07}. Let
$\mu^s_{\mathcal{T}_{x}}$ denote the (finite-volume) Gibbs measure
in which the parent of $x$ has its spin fixed to $s$ and the
configuration on the bottom boundary of $\mathcal{T}_{x}$ (i.e.,
on $\partial\mathcal{T}_{x}\setminus \{\text{parent  of}\ x\}$) is
specified by $\mathcal{T}$ (see
\cite{MSW07,Mos01,Sly}).

For two measures $\mu_1$ and $\mu_1$ on $\Omega$, $||\mu _1-\mu
_2||_x$ denotes the variation distance between the projections of
$\mu_1$ and $\mu_1$ onto the spin at $x$, i.e.
$$
||\mu _1-\mu _2||_x=\frac{1}{2}\sum _{i\in \Phi } \left|\mu
_1(\sigma (x)=i)-\mu _2(\sigma (x)=i)\right|.
$$
Let $\eta^{x,s}$ be the configuration $\eta$ with the spin at $x$
set to $s$.
\begin{defin}\cite[Definition 3.1]{MSW07}\label{def-extremality}
For a collection of Gibbs distributions
$\{\mu^s_{\mathcal{T}_{x}}\}$, define the quantities
$\kappa\equiv\kappa(\{\mu^s_{\mathcal{T}_{x}}\})$ and
$\gamma\equiv\gamma(\{\mu^s_{\mathcal{T}_{x}}\})$ by
\begin{enumerate}
    \item $\kappa =\sup\limits_{z\in\Gamma^k}\max\limits_{z,s,s'}\|\mu _{T_z}^s-\mu
    _{T_z}^{s'}\|_z$
    \item $\gamma =\sup\limits_{A\subset\Gamma^k}\max _{z,s,s'}\|\mu _A^{\eta ^{y,s}}-\mu _A^{\eta
    ^{y,s'}}\|_z$, where the maximum is taken over all boundary conditions $\eta$, all sites $y\in \partial A$, all neighbors
$x\in A$ of $y$, and all spins $s,s'\in \{-1,0,1\}$.
\end{enumerate}
\end{defin}

It is known \cite[Theorem 9.3]{MSW07} that a sufficient condition for extremality of the translation-invariant Gibbs measure is the following
inequality:
\begin{equation}\label{exramality1}
k \kappa \gamma < 1.
\end{equation}
Note that $\kappa$ has the particularly simple form $\kappa
=\frac{1}{2}\max\limits _{i,j}\left\{\sum\limits_{l=1}^3
\left|H_{i,l}-H_{j,l}\right|\right\}$.

From the Eq. \eqref{matHS2}, we obtain
\small\begin{eqnarray*}
\kappa&=&\frac{1}{2(1+3\theta ^2+4\theta ^3+3\theta ^4+\theta ^6)}\max
\left[\left\{(-1+\theta )^2(1+\theta )^2,2(-1+\theta )^2(1+\theta
)^2,(-1+\theta )^2(1+\theta
)^2\right\}\right]\\
&=&\frac{\left(1+\theta ^2\right)\left(-1+\theta
^2\right)^2}{(1+3\theta ^2+4\theta ^3+3\theta ^4+\theta ^6)},
\end{eqnarray*}
here we have used that $k=2$.

\normalsize

From the definition \ref{def-extremality}, it is easy that
$\gamma=\kappa$. So, the extramality condition \eqref{exramality1} reduces to
\begin{equation*}\label{Kappa}
2\kappa \gamma -1=2\left(\frac{\left(1+\theta
^2\right)\left(-1+\theta ^2\right)^2}{(1+3\theta ^2+4\theta
^3+3\theta ^4+\theta ^6)}\right)^2-1<0.
\end{equation*}
One can see that $1+3 \theta ^2+4 \theta ^3+3 \theta ^4+\theta
^6=(1+\theta^2)^3+4\theta^3$. Therefore,
\begin{eqnarray}\label{Kappa1}
2\kappa \gamma -1&=2&\left(\frac{\left(-1+\theta ^2\right)^2
\left(1+\theta ^2\right)}{(1+\theta ^2)^3+4\theta
^3}\right)^2-1\\\nonumber &=&\left(\frac{\sqrt{2}\left(-1+\theta
^2\right)^2 \left(1+\theta ^2\right)}{(1+\theta ^2)^3+4\theta
^3}-1\right)\left(\frac{\sqrt{2}\left(-1+\theta ^2\right)^2
\left(1+\theta ^2\right)}{(1+\theta ^2)^3+4\theta^3}+1\right)<0.
\end{eqnarray}
It is obvious that  \eqref{Kappa1} is equivalent to
\begin{equation}\label{Kappa2}
\frac{\sqrt{2}\left(-1+\theta ^2\right)^2 \left(1+\theta
^2\right)}{(1+\theta ^2)^3+4\theta ^3}-1<0.
\end{equation}
From \eqref{Kappa2}, one gets
$$
\sqrt{2} \left(-1+\theta^2\right)^2
\left(1+\theta^2\right)-\left(4
\theta^3+\left(1+\theta^2\right)^3\right)<0.
$$
After some long and tedious algebraic operations, the last one reduces to
$$
(-1+\sqrt{2})+\left(-3-\sqrt{2}\right) \theta^2-4
\theta^3+\left(-3-\sqrt{2}\right)
\theta^4+\left(-1+\sqrt{2}\right)\theta^6<0.
$$
and after some algebraic operations, finally we arrive at
\begin{equation}\label{V(teta2)}
(\sqrt{2}-1)\vartheta^3-(4\sqrt{2}+\sqrt{3}-3)\vartheta-4<0,
\end{equation}
where $\vartheta=\theta + \frac{1}{\theta}.$

Using Cardano's formula and Mathematica \cite{Wolfram}, the solution of \eqref{V(teta2)} is $\vartheta<\vartheta_{crt}$, where
$\vartheta_{crt}\approx 3.634$. Therefore, from $\theta +
\frac{1}{\theta}<3.634$, one finds
\begin{equation}\label{f(teta)}
f(\theta)=\theta^2-3.634\theta+1<0.
\end{equation}
So, the solution of the last inequality is $\theta\in(\tilde\theta_1,\tilde\theta_2)$, where
\begin{equation}\label{tta}
\tilde\theta_{1,2}=\frac{1}{2}\big(\vartheta_{crt}\pm\sqrt{\vartheta_{crt}^2-4}\big).
\end{equation}
Note that
$\tilde\theta_1\approx 0.299934,\tilde\theta_2\approx 3.33407$.

Consequently, we get the following result.

\begin{thm}\label{extr} If $\theta\in(\tilde\theta_1,\tilde\theta_2)$, then the disordered phase is extreme.
\end{thm}
%

Now, it is natural to ask: if $\theta\notin(\tilde\theta_1,\tilde\theta_2)$, is the disordered phase extreme?
To respond the raised question, we are going to determine the regions of the
parameters in which the disordered phase is not extreme in the set
of all Gibbs measures (including the non-translation invariant
ones). It is known that a sufficient condition (Kesten-Stigum
condition \cite{Kesten-Stigum-1966}) for non-extremality of a
Gibbs measure $\mu$ corresponding to the matrix $\mathbb{H}$ on a
Cayley tree of order $k\geq 1$ is that $k|\lambda^2_{\max}| > 1$,
where $|\lambda_{\max}|$ is the second largest (in absolute value)
eigenvalue of $\mathbb{H}$ given in \eqref{matHS2}. Furthermore,
we are going to employ this condition to determine the
non-extremity of the disordered phase.

One can show that the set of the eigenvalues of the stochastic
$\mathbb{H}$ given in \eqref{matHS2} are the following ones:

\begin{equation*}\label{eigen-H} \left\{0,\frac{\left(-1+\theta
^2\right)^2 \left(1+\theta ^2\right)}{1+3 \theta ^2+4 \theta ^3+3
\theta ^4+\theta ^6},1\right\}.
\end{equation*}
Due to $1+3 \theta ^2+4 \theta ^3+3 \theta ^4+\theta
^6=(1+\theta^2)^3+4\theta^3$, the above condition can be rewritten as
\begin{eqnarray}\label{V(teta)}
V(\theta)&=&2\lambda^{2}_{\max}-1 \nonumber\\
&=&2\left(\frac{\left(-1+\theta ^2\right)^2 \left(1+\theta
^2\right)}{(1+\theta ^2)^3+4\theta ^3}\right)^2-1 \nonumber\\
&=&\left(\frac{\sqrt{2}\left(-1+\theta ^2\right)^2 \left(1+\theta
^2\right)}{(1+\theta ^2)^3+4\theta
^3}-1\right)\left(\frac{\sqrt{2}\left(-1+\theta ^2\right)^2
\left(1+\theta ^2\right)}{(1+\theta ^2)^3+4\theta^3}+1\right)>0.
\end{eqnarray}
Obviously \eqref{V(teta)} is equivalent to
\begin{equation}\label{1V(teta1)}
\frac{\sqrt{2}\left(-1+\theta ^2\right)^2 \left(1+\theta
^2\right)}{(1+\theta ^2)^3+4\theta ^3}-1>0.
\end{equation}
From \eqref{1V(teta1)}, one finds
$$
\sqrt{2} \left(-1+\theta^2\right)^2
\left(1+\theta^2\right)-\left(4
\theta^3+\left(1+\theta^2\right)^3\right)>0.
$$
After some long and tedious algebraic operations, we get
$$
(-1+\sqrt{2})+\left(-3-\sqrt{2}\right) \theta^2-4
\theta^3+\left(-3-\sqrt{2}\right)
\theta^4+\left(-1+\sqrt{2}\right)\theta^6>0.
$$
Hence, by denotiing $\vartheta=\theta + \frac{1}{\theta}$ the last one is reduced to
\begin{equation}\label{1V(teta2)}
(\sqrt{2}-1)\vartheta^3-(4\sqrt{2}+\sqrt{3}-3)\vartheta-4>0.
\end{equation}

One can observe that \eqref{1V(teta2)} and \eqref{V(teta2)} are complementary inequalities, therefore, the solution of \eqref{1V(teta2)} (in terms of $\theta$) is $(0,\tilde\theta_1)\cup (\theta_2,\infty)$.
Now, keeping in mind Theorem \ref{extr} we obtain the following important result.

 \begin{thm}\label{exrt1} Let $J\neq 0$ and consider the Ising model with mixed spin $(1,1/2)$ on a Cayley tree of order two. Then There exist $\tilde\theta_1$ and $\tilde\theta_2$ (see \eqref{tta}) such that the disordered phase $\mu_0$ is extreme if and only if
    $$
    \theta\in (\tilde\theta_1, \tilde\theta_2);
    $$
 where $\theta=\exp\{J\beta/2\}$.
 \end{thm}

 \begin{rem} We notice that, for the Ising model on the Cayley tree, the
disordered phase is extreme, if and only if $\exp\{J\beta\}<1/\sqrt{k}$ \cite{Bleher1990,Ioffe-1996}.
On the other hand, the Ising model with mixed spin $(1,1/2)$ also has like Potts kind of behavior. However, for the Potts model
on the Cayley tree the extremality of the disordered phase has been investigated in \cite{KR17,Mos01,Haydarov2016}. It turns out that there are two critical values  $\theta_1$ and $\theta_2$ such that if $\theta<\theta_1$ then the disordered phase is extreme, if $\theta>\theta_2$ then it is not extreme, but in the region $(\theta_1,\theta_2)$ its the extremality is still an open problem. We stress that in the considered setting,
the extremity condition is similar to the Ising model.
 \end{rem}

%

\section*{Acknowledgements.} The first author (H. A) thanks ICTP for
providing financial support and all facilities. He is also grateful
to the Simons Foundation and IIE for their support.


\end{document}